\documentclass[letterpaper, 10 pt, conference]{ieeeconf}
\IEEEoverridecommandlockouts

\usepackage{amsmath, amsthm, amssymb, amsfonts, mathtools, xargs, tensor, units, cite, stmaryrd, mathrsfs, algpseudocode, comment, graphicx}
\usepackage[unicode=true, bookmarks=true, bookmarksnumbered=false, bookmarksopen=false, breaklinks=false, pdfborder={0 0 1}, backref=false, colorlinks=false, hidelinks]{hyperref}
\makeatletter
\let\NAT@parse\undefined
\makeatother
\newcommand{\dif}{\mathop{}\!\mathrm{d}}
\usepackage[colorinlistoftodos]{todonotes}
\theoremstyle{defn}

\theoremstyle{assumption}
\newtheorem{assumption}{Assumption}

\theoremstyle{lem}

\theoremstyle{thm}
\newtheorem{thm}{Theorem}

\theoremstyle{remark}
\newtheorem{remark}{Remark}

\theoremstyle{cor}

\begin{document}
	\title{Online inverse reinforcement learning with limited data}
	
	\author{Ryan Self, S M Nahid Mahmud, Katrine Hareland and Rushikesh Kamalapurkar\thanks{The authors are with the School of Mechanical and Aerospace Engineering, Oklahoma State University, Stillwater, OK, USA. {\tt\small \{rself,nahid.mahmud,katrine.hareland, rushikesh.kamalapurkar\}@okstate.edu}. This research was supported, in part, by the National Science Foundation (NSF) under award number 1925147. Any opinions, findings, conclusions, or recommendations detailed in this article are those of the author(s), and do not necessarily reflect the views of the sponsoring agencies.}}
	\maketitle
	\begin{abstract} 
		This paper addresses the problem of online inverse reinforcement learning for systems with limited data and uncertain dynamics. In the developed approach, the state and control trajectories are recorded online by observing an agent perform a task, and reward function estimation is performed in real-time using a novel inverse reinforcement learning approach. Parameter estimation is performed concurrently to help compensate for uncertainties in the agent's dynamics. Data insufficiency is resolved by developing a data-driven update law to estimate the optimal feedback controller. The estimated controller can then be queried to artificially create additional data to drive reward function estimation.
	\end{abstract}
	\section{Introduction}
	
	Based on the premise that the most succinct representation of the behavior of an entity is its reward structure \cite{SCC.Ng.Russell2000}, this paper aims to recover the reward (or cost) function of a demonstrator by monitoring its state and control trajectories. Reward function estimation is performed in the presence of modeling uncertainties for situations with limited data via inverse reinforcement learning (IRL) \cite{SCC.Ng.Russell2000,SCC.Russell1998}.
	
	While IRL in an \emph{offline} setting has a rich history of literature \cite{SCC.Ng.Russell2000,SCC.Russell1998,SCC.Abbeel.Ng2004,SCC.Abbeel.Ng2005,SCC.Ratliff.Bagnell.ea2006,SCC.Ziebart.Maas.ea2008,SCC.Zhou.Bloem.ea2018,SCC.Levine.Popovic.ea2010,SCC.Neu.Szepesvari2007,SCC.Syed.Schapire2008,SCC.Levine.Popovic.ea2011}, little work has been done to address IRL in an online setting. One reason for this is the limited data provided by a single demonstration. 
	
	Preliminary results on online IRL are available for linear systems, in results such as \cite{SCC.Kamalapurkar2018} and \cite{SCC.Molloy.Ford.ea2018}, and for nonlinear systems, in results such as \cite{SCC.Self.Harlan.ea2019a} and \cite{SCC.Self.Abudia.ea2020}. However, \cite{SCC.Kamalapurkar2018} and \cite{SCC.Self.Harlan.ea2019a} exploit access to demonstrator's feedback policy, \cite{SCC.Molloy.Ford.ea2018} requires exact model knowledge, and \cite{SCC.Self.Abudia.ea2020} exploits identical disturbances to provide sufficient excitation. The main contribution of this paper is the development of a novel method for reward function estimation for an agent in situations where estimation of the demonstrator's optimal feedback law is less data-intensive than direct estimation of its reward function.
	
	The novelty in the technique developed in this paper is a recursive model-based IRL approach which facilitates the use of off-trajectory state-action pairs. A majority of IRL methods are trajectory-driven and model-free. As a result, the trajectories need to be sufficiently information-rich for reward function estimation. The technique developed in this paper is model-based, and as a result, once a model is learned, arbitrary state-action pairs can be used for IRL as long as the action is the optimal action for that state. In \cite{SCC.Kamalapurkar2018} and \cite{SCC.Self.Harlan.ea2019a}, the off-trajectory state-action pairs are generated under the assumption that the learner either knows the demonstrator's optimal feedback law or can query the demonstrator to find out what the optimal action would be at a given off-trajectory state. In this paper, we develop a novel IRL approach that relaxes the aforementioned assumption.
	
	The key idea in this paper is to estimate the optimal feedback controller of the agent online, and use that estimate to artificially create off-trajectory data to drive reward function estimation. In the authors' previous work \cite{SCC.Self.Harlan.ea2019a}, reward function estimation is performed directly using the agents observed trajectories. Instead, in this paper, the trajectory information is used to estimate the optimal feedback controller. This controller is parameterized as a neural network and estimated using a concurrent learning update law. The estimated controller is simultaneously queried to create off-trajectory data which is then used for reward function estimation via IRL. Since the optimal controller is estimated using a neural network, the controller can be estimated independent of the modeling uncertainty. In the developed approach, parameter estimation and two update laws for estimation of the optimal feedback controller and reward function are utilized simultaneously, to achieve uniform ultimate boundedness of the unknown reward function weights.
	
	The paper is organized as follows: Section \ref{sec:Notation} explains the notation used throughout the paper. Section \ref{sec:Problem Formulation} details the problem formulation. Section \ref{sec: Estimate Control} shows how to estimate the optimal controller. Section \ref{sec:IRLNoState} explains the IRL algorithm. Section \ref{sec: Sim} shows a simulation example and Section \ref{sec: Conclusion} concludes the paper. 
	
	\section{Notation}\label{sec:Notation}
	The notation $\mathbb{R}^{n}$ represents the $n-$dimensional Euclidean space, and the elements of $\mathbb{R}^{n}$ are interpreted as column vectors, where $\left(\cdot\right)^{T}$ denotes the vector transpose operator. The set of positive integers excluding 0 is denoted by $\mathbb{N}$. For $a\in\mathbb{R},$ $\mathbb{R}_{\geq a}$ denotes the interval $\left[a,\infty\right)$, and $\mathbb{R}_{>a}$ denotes the interval $\left(a,\infty\right)$. If $a\in\mathbb{R}^{m}$ and $b\in\mathbb{R}^{n}$, then $\left[a;b\right]$ denotes the concatenated vector $\begin{bmatrix}a\\
	b
	\end{bmatrix}\in\mathbb{R}^{m+n}$. The notations $\text{I}_{n}$ and $0_{n}$ denote the $n\times n$ identity matrix and the zero element of $\mathbb{R}^{n}$, respectively. Whenever it is clear from the context, the subscript $n$ is suppressed.
	
	\section{Problem Formulation} \label{sec:Problem Formulation}
	Consider an agent with the following dynamics
	\begin{align}
	\dot{x}&=f(x,u),\label{eq:Nonlinear Problem FormulationNoState}
	\end{align}
	where $x:\mathbb{R}_{\geq T_0}\rightarrow \mathbb{R}^n$ is the state, $u:\mathbb{R}_{\geq T_0}\rightarrow \mathbb{R}^m$ is the control, and $f:\mathbb{R}^n\times\mathbb{R}^m\to\mathbb{R}^n$ is a continuously differentiable function. 
%
	
	The agent under observation is using the policy which minimizes the following performance index
	\begin{equation}
	J(x_0,u(\cdot)) = \int_{T_0}^{\infty}r(x(t;x_0,u(\cdot)),u(t))\dif t,\label{eq:reward}
	\end{equation}
	where $x(\cdot;x_0,u(\cdot))$ is the trajectory of the agent generated by the optimal control signal $u(\cdot)$ that minimizes the performance index in \eqref{eq:reward} starting from the initial condition $x_0$ and beginning at time $T_0$. The main objective of the paper is to estimate the unknown reward function, $r$, using input-state pairs.
	
	The following assumptions are used throughout the rest of this paper.
	\begin{assumption} \label{assume: reward}
		The unknown reward function $r$ is quadratic in the control, i.e.,
		\begin{equation}
		r(x,u) = Q(x)+u^TRu,
		\end{equation}
		where $R\in \mathbb{R}^{m\times m}$ is a positive definite matrix, such that $R=\mathrm{diag}([r_1,\cdots,r_m])$.
	\end{assumption} 
	\cite{SCC.Hornik.Stinchcombe.ea1985,SCC.Hornik1991} The continuous function $Q$ can be represented using a neural network as $Q(x) =(W^*_Q)^T\sigma_{Q}(x)+\epsilon_Q(x)$, where $ W_{Q}^{*}\coloneqq\left[q_{1},\cdots,q_{L}\right]^{T} $ are ideal reward function weights, $\sigma_{Q}: \mathbb{R}^{n} \to \mathbb{R}^L$ are known continuously differentiable features, and $\epsilon_Q:\mathbb{R}^n\to \mathbb{R}$ is the approximation error.
	\begin{assumption}\label{assume: affine}
		The dynamics for the agent are affine in control and can be expressed as
		\begin{equation}
		\dot{x}=f^o(x,u)+\theta^T\sigma(x,u)+\epsilon(x,u), \label{eq:NLSPre Integral Form}
		\end{equation}
		where $f^o:\mathbb{R}^n\times\mathbb{R}^m\rightarrow \mathbb{R}^n$ denotes the continuously differentiable nominal dynamics, $\theta^T\sigma$ is a parameterized estimate of the uncertain part of the dynamics, where $\theta \in \mathbb{R}^{p\times n}$ is a matrix of unknown constant parameters and $\sigma:\mathbb{R}^{n}\times \mathbb{R}^m\to \mathbb{R}^p$ are known continuously differentiable features, and $\epsilon:\mathbb{R}^n\times\mathbb{R}^m\to\mathbb{R}^n$ denotes the function approximation error.
	\end{assumption}

	Under the premise that the observed agent makes optimal decisions, the state and control trajectories, $x(\cdot)$ and $u(\cdot)$, satisfy the Hamilton-Jacobi-Bellman (HJB) \cite{SCC.Liberzon2012} equation
	\begin{equation}
	H\left(x\left(t\right),\Big(\left[\nabla_{x}V^{*}\right]\left(x\left(t\right)\right)\Big)^{T},u\left(t\right)\right)=0,\forall t\in \mathbb{R}_{\geq T_0},\label{eq:inverse HJB}
	\end{equation}
	where the unknown optimal value function is $ V^{*}:\mathbb{R}^{n}\to \mathbb{R} $ and $H:\mathbb{R}^n\times \mathbb{R}^n \times \mathbb{R}^m \to \mathbb{R}$ is the Hamiltonian, defined as $H(x,p,u):=p^Tf(x,u)+r(x,u)$. The goal of IRL is to estimate the reward function, $ r $. 
	
	To aid in the estimation of the reward function, let $ \hat{V}:\mathbb{R}^{n}\times \mathbb{R}^{P} \to \mathbb{R}$, $ \left(x,\hat{W}_{V}\right)\mapsto \hat{W}_{V}^{T}\sigma_{V}\left(x\right)$ be a parameterized estimate of the optimal value function $V^{*}$, where $ \hat{W}_{V}\in \mathbb{R}^{P} $ are the estimates of the ideal value function weights $W_V^*$ and $ \sigma_{V}:\mathbb{R}^{n}\to\mathbb{R}^{P} $ are known continuously differentiable features. Let $\epsilon_V:\mathbb{R}^n\to\mathbb{R}$, defined as $\epsilon_V(x)=V^*(x)-\left(W_V^*\right)^T\sigma_V(x)$, be the resulting approximation error. Using $ \hat{W}_{V} $, $ \hat{W}_{Q} $, and $ \hat{W}_{R} $, which are the estimates of $ W_{V}^{*} $, $ W_{Q}^{*} $, and $ W_{R}^{*}\coloneqq\left[r_{1},\cdots,r_{m}\right]^{T} $, respectively, in \eqref{eq:inverse HJB}, the inverse Bellman error $ \delta:\mathbb{R}^{n}\times \mathbb{R}^{m}\times\mathbb{R}^{L+P+m}\to\mathbb{R} $ is obtained as
	\begin{align}
	\delta\left({x},u,\hat{W}\right)=&\hat{W}_{V}^{T}\Big(\left[\nabla_{x}\sigma_{V}\right]\left({x}\right) \Big) f(x,u)+\hat{W}_{Q}^{T}\sigma_{Q}\left({x}\right)\nonumber\\
	&+\hat{W}_{R}^{T}\sigma_{u}\left(u\right),\label{eq: IBE}
	\end{align}where $ \sigma_{u}\left(u\right)\coloneqq\left[u_{1}^{2},\cdots,u_{m}^{2}\right]$.
	
	For brevity of presentation, it is assumed that a parameter estimator that satisfies the following properties is available. For examples of such parameter estimates, see \cite{SCC.Kamalapurkar2017,SCC.Self.Harlan.ea2019a}.
	\begin{assumption}\label{assume: StateParamUUB}
		\cite[Assumption 2]{SCC.Kamalapurkar.Walters.ea2016} A compact set $\Theta\subset\mathbb{R}^p$ such that $\theta\in\Theta$ is known a priori. The estimate $\hat{\theta}:\mathbb{R}_{\geq T_0}\to\mathbb{R}^p$ are updated based on a switched update law of the form
		\begin{equation*}
			\dot{\hat{\theta}}=f_{\theta_s}(\hat{\theta}(t),t),
		\end{equation*}
		$\hat{\theta}(T_0)=\hat{\theta}_0\in\Theta,$ where $s\in\mathbb{N}$ denotes the switching index and $\{f_{\theta_s}:\mathbb{R}^p\times\mathbb{R}_{\geq T_0}\to\mathbb{R}^p\}_{s\in\mathbb{N}}$ denotes the family of continuously differentiable functions. The dynamics of the parameter estimation error $\tilde{\theta}:\mathbb{R}_{\geq T_0}\to\mathbb{R}^p$, defined as $\tilde{\theta}(t):=\theta-\hat{\theta}(t)$, can be expressed as $\dot{\hat{\theta}}(t)=f_{\theta_s}\left(\theta-\tilde{\theta}(t),t\right)$. Furthermore, there exists a continuously differentiable function $V_\theta:\mathbb{R}^p\times\mathbb{R}_{\geq T_0} \to\mathbb{R}_{\geq0}$ that satisfies 
		\begin{equation*}
			\underline{\nu}_\theta\left(\left\Vert\tilde{\theta}\right\Vert\right)\leq V_\theta\left(\tilde{\theta},t\right)\leq\overline{\nu}_\theta\left(\left\Vert\tilde{\theta}\right\Vert\right),
		\end{equation*}
		and
		\begin{multline*}
			\left(\left[\nabla_{\tilde{\theta}}V_\theta\right]\left(\tilde{\theta},t\right)\right)\left(-f_{\theta_s}\left(\theta-\tilde{\theta},t\right)\right)+\frac{\partial V_\theta\left(\tilde{\theta},t\right)}{\partial t}\\\leq-K\left\Vert\tilde{\theta}\right\Vert^2+D\left\Vert\tilde{\theta}\right\Vert,
		\end{multline*}
		for all $s\in\mathbb{N},t\in\mathbb{R}_{\geq T_0},$ and $\tilde{\theta}\in\mathbb{R}^p$, where $\underline{\nu}_\theta,\overline{\nu}_\theta:\mathbb{R}_{\geq 0}\to\mathbb{R}_{\geq 0}$ are class $\mathcal{K}$ functions, $K\in\mathbb{R}_{>0}$ is an adjustable parameter, and $D\in\mathbb{R}_{>0}$ is a positive constant.
	\end{assumption}
	Utilizing parameter estimates from Assumption \ref{assume: StateParamUUB}, \eqref{eq: IBE} can be updated and expressed as
	\begin{align}
	\delta^{\prime}\left({x},u,\hat{W},\hat{\theta}\right)\!=&\hat{W}_{V}^{T}\Big(\!\left[\nabla_{x}\sigma_{V}\right]\left({x}\right) \Big)  \hat{Y}(x,u,\hat{\theta})+\hat{W}_{Q}^{T}\sigma_{Q}\left({x}\right)\nonumber\\
	&+\hat{W}_{R}^{T}\sigma_{u}\left(u\right),\label{eq: IBE_Prime}
	\end{align} where $\hat{Y}(x,u,\hat{\theta}):=f^o(x,u)+\hat{\theta}^T\sigma(x,u)$ and $\hat{\theta}$ are estimates of unknown parameters. Rearranging, \eqref{eq: IBE_Prime} becomes \begin{equation}
	\delta^{\prime}\left({x},u,\hat{W}^{\prime},\hat{\theta}\right)=\left(\hat{W}^{\prime}\right)^{T}\sigma^{\prime}\left({x},u,\hat{\theta}\right),\label{eq:inverse BE}
	\end{equation}where $ \hat{W}^{\prime} \coloneqq \left[\hat{W}_{V};\hat{W}_{Q};\hat{W}_{R}\right] $ and $ \sigma^{\prime}\left({x},u,\hat{\theta}\right)\coloneqq\left[\Big(\left[\nabla_{x}\sigma_{V}\right]\left({x}\right) \Big)\hat{Y}(x,u,\hat{\theta});\sigma_{Q}\left({x}\right);\sigma_{u}\left(u\right)\right] $.
	
	In the following, the parameter estimator is executed synchronously in with IRL and in real-time.

	\section{Optimal Controller Estimation}\label{sec: Estimate Control}
	Since a large majority of optimal control problems are aimed at driving the state to a set-point or an error signal to zero, information content of the state and control trajectories can quickly decay to zero rendering them unable to provide usable data. More specifically, once the states converge, newer data points from the agent's trajectories will simply provide zero, or near-zero, values for both the states (or errors) and the controls. As a result, the reward function estimate may never converge. Motivated by the observation that knowledge of the optimal controller can be leveraged to artificially create additional data to drive IRL, this section develops a process for finding an estimate of the optimal controller. 
	
	\subsection{Controller Estimation Formulation}
	Provided Assumptions \ref{assume: reward} and \ref{assume: affine} are satisfied, the closed-form nonlinear optimal controller corresponding to the reward structure in (\ref{eq:reward}) is\begin{equation}\label{optimal controller}
	u^*(x)=-\frac{1}{2}R^{-1}\Big(\left[\nabla_uf\right](x)\Big)^T\Big(\left[\nabla_{x}V^*\right](x)\Big)^T,
	\end{equation}
	where $u^*:=[u_1,u_2,\cdots,u_m]^T$ and $\Big(\left[\nabla_uf\right]\left(x\right)\Big)$ is found from $f(x,u)$ in \eqref{eq:Nonlinear Problem FormulationNoState}. To promote estimation, $u^*$ will be represented as
	\begin{equation}
	u^*(x)=-\left({W}_{u}^*\right)^{T}\sigma_{u}\left(x\right)+\epsilon_u(x),\label{optimal controller estimated}
	\end{equation}
	where $W^*_u\in\mathbb{R}^{K\times m}$ is a matrix of unknown ideal constant parameters, $ \sigma_{u}:\mathbb{R}^{n}\to\mathbb{R}^{K} $ are known continuously differentiable features, and $\epsilon_u:\mathbb{R}^n\to\mathbb{R}^m$ is the resulting approximation error.
	
	Collecting state and control signals over time instances, $t_1,t_2,\cdots,t_M$, stored in a history stack, denoted as $\mathcal{H}^{u}$, \eqref{optimal controller estimated} can be formulated into the matrix form
	\begin{equation}
	-\Sigma_u-\Sigma_{\sigma}\hat{W}_u={\Sigma}_{\sigma}\tilde{W}_u-{\Delta}_u,\label{optimal controller matrix}
	\end{equation}
	where $\Sigma_u:=[u^T(t_1);u^T(t_2);\cdots;u^T(t_M)]$, ${\Sigma}_{\sigma}:=[{\sigma}^T_u({x}(t_1));{\sigma}^T_u({x}(t_2));\cdots;{\sigma}^T_u({x}(t_M))]$, and ${\Delta}_u:=[{\epsilon}^T_u({x}(t_1));{\epsilon}^T_u({x}(t_2));\cdots;{\epsilon}^T_u({x}(t_M))]$. The weight estimation error is defined as $\tilde{W}_u=W_u^*-\hat{W}_u$, where $\hat{W}_u$ is the estimate of $W^*_u$.
	
	Using \eqref{optimal controller matrix}, a recursive least-squares update law to estimate the unknown weights is designed as 
	\begin{equation}
	\dot{\hat{W}}_u=\alpha_u\Gamma_u {\Sigma}_{\sigma}^T\left(-\Sigma_u-{\Sigma}_{\sigma}\hat{W}_u\right).\label{eq: WK Error Dyns}
	\end{equation}
	where $\alpha_u\in\mathbb{R}_{>0}$ is a constant adaptation gain, and $\Gamma_u:\mathbb{R}_{\geq0}\to\mathbb{R}^{K\times K}$ is the least-squares gain updated using the update law
	\begin{equation}
	\dot{\Gamma}_u=\beta_u\Gamma_u-\alpha_u\Gamma_u{\Sigma}_{\sigma}^T{\Sigma}_\sigma\Gamma_u.\label{eq: Gamma WK}
	\end{equation}
	where $\beta_u\in \mathbb{R}_{>0}$ is the forgetting factor.
	
	\subsection{Analysis}
	The time-varying history stack, $\mathcal{H}^{u}$, is called full rank, uniformly in $t$, if there exists a $\underline{k}>0$ such that $\forall t\in\mathbb{R}_{\geq T_0}$,
	\begin{equation}
	0<\underline{k}<\lambda_{\min}\left\{{\Sigma}_\sigma^T(t){\Sigma}_\sigma(t)\right\}.\label{def: KFull Rank}
	\end{equation}
	
	Using arguments similar to \cite[Corollary 4.3.2]{SCC.Ioannou.Sun1996}, it can be shown that if $\lambda_{\min}\left\{ \Gamma_u^{-1}\left(0\right)\right\} >0$, and if $\mathcal{H}^{u}$ is full rank, uniformly in $t$, then the least squares gain matrix satisfies
	\begin{equation}
	\underline{\Gamma}_u\text{I}_{K}\leq\Gamma_u\left(t\right)\leq\overline{\Gamma}_u\text{I}_{K},\label{eq:KGammaBound}
	\end{equation}
	where $\underline{\Gamma}_u$ and $\overline{\Gamma}_u$ are positive
	constants.
	
	To facilitate the following analysis, using \eqref{optimal controller matrix} and \eqref{eq: WK Error Dyns}, the dynamics for the weight estimation error can be described by
	\begin{equation}
	\dot{\tilde{W}}_u=-\alpha_u\Gamma_u{\Sigma}_\sigma^T\left({\Sigma}_\sigma\tilde{W}_u-\Delta_u\right).\label{eq:Wkerrordyns}
	\end{equation} 
	
	\begin{thm}
		If $\mathcal{H}^{u}$ is full rank, uniformly in $t$, then $t \mapsto \tilde{W}_u\left(t\right) $ is uniformly ultimately bounded.
	\end{thm}	
	\begin{proof}
		Consider the following positive definite candidate Lyapunov function 
		\begin{equation}
		\label{eqn:V_estimation}
		V_u(\tilde{W}_u,t) = \text{tr}(\tilde{W}_u^T \Gamma_u^{-1}(t) \tilde{W}_u),
		\end{equation}
		Using the bounds in \eqref{eq:KGammaBound}, the candidate Lyapunov function satisfies 
		\begin{equation}
		\frac{1}{\overline{\Gamma}_u}\left\Vert \tilde{W}_u\right\Vert ^{2}\leq V_u\left(\tilde{W}_u,t\right)\leq\frac{1}{\underline{\Gamma}_u}\left\Vert \tilde{W}_u\right\Vert ^{2}.
		\end{equation}
		Taking the time derivative of (\ref{eqn:V_estimation}), and using \eqref{eq: Gamma WK} and \eqref{eq:Wkerrordyns}, along with the identity $\dot{\Gamma}_u^{-1}=-\Gamma_u^{-1}\dot{\Gamma}_u\Gamma_u^{-1}$, after simplifying yields
		%
		\begin{multline}
		\dot{V}_u(\tilde{W}_u,t)= -\alpha_u\text{tr}( \tilde{W}_u^T  {\Sigma}_\sigma^T{\Sigma}_\sigma\tilde{W}_u)\\+2\alpha_u\text{tr}(\tilde{W}^T_u{\Sigma}_\sigma^T\Delta_u) 
		-\beta_u\text{tr}( \tilde{W}_u^T \Gamma_u^{-1}(t)    {\tilde{W}}_u).
		\end{multline}
		Using the Cauchy-Schwartz inequality, and bounds in (\ref{def: KFull Rank}) and (\ref{eq:KGammaBound}), $\dot{V}_u$ can be bounded by 
		\begin{multline}
		\label{eq:Dot_V_estimation_bound2_theta}
		\dot{V}_u(\tilde{W}_u,t)\leq -\left(\alpha_u\underline{k} + \frac{\beta_u}{\overline{\Gamma}_u}\right)\left\Vert\tilde{W}_u\right\Vert^2 \\+2\alpha_u\left\Vert\tilde{W}_u\right\Vert\left\Vert {\Sigma}_\sigma\right\Vert\left\Vert\Delta_u\right\Vert.
		\end{multline}
		
		Since the states and controls are both bounded, $\left\Vert {\Sigma}_\sigma\right\Vert$ and $\Vert\Delta_{u}\Vert$ are bounded above. The upper bounds are defined as $
		\overline{\Sigma}_\sigma$ and
		$\overline{\Delta}_u.$
		Using these upper bounds and Young's Inequality, $\dot{V}_u$ becomes
		\begin{equation}
		\label{eq: Differential V}
		\dot{V}_u(\tilde{W}_u,t)\leq -AV_u\left(\tilde{W}_u,t\right) +B,
		\end{equation}
		where $A$ and $B$ are defined as
		\begin{equation}
		A:=\frac{\underline{\Gamma}_u}{2}\left(\alpha_{u}\underline{k} + \frac{\beta_u}{\overline{\Gamma}_u}\right),
		\end{equation}
		and
		\begin{equation}
		B:=\frac{2(\alpha_u\overline{\Sigma}_\sigma\overline{\Delta}_u)^2}{(\alpha_{u}\underline{k}+\nicefrac{\beta_{u}}{\overline{\Gamma}_u})}.
		\end{equation}
		Finding the solution of $\eqref{eq: Differential V}$ yields
		\begin{equation}\label{eq: V Equation}
		{V}_{u}(t)\leq {V}_{u_0}e^{-A\left(t-T_0\right)}+\frac{B}{A},
		\end{equation}
		where ${V}_{u_0}\geq\left\Vert V_u\left(\tilde{W}_u\left(T_{0}\right),T_{0}\right)\right\Vert$. It can be concluded that \begin{equation}
		\lim\limits_{t\to\infty}{V}_u(t)\leq\frac{B}{A}.\end{equation} It can further be concluded that $\tilde{W}_u$ decays exponentially, such that \begin{equation}\lim\limits_{t\to\infty}\left\Vert\tilde{W}_u\left(t\right)\right\Vert\leq \sqrt{\overline{\Gamma}_u\frac{{B}}{A}}.\end{equation}
	\end{proof}
	
	\section{Inverse Reinforcement Learning}\label{sec:IRLNoState}
	In this section, the optimal feedback estimator developed in this previous section is utilized to create a data-set of estimated near-optimal state-action pairs to drive IRL.	
	
	\subsection{Utilizing Control and Parameter Estimates}
	Consider a time instance, $t_i$. For each time $t_i$, select an arbitrary state, denoted by $x_i$, and let $\hat{u}_i:=\hat{W}_u^T(t_i)\sigma_u(x_i)$ be the estimate of the optimal controller $u_i^*$ at state $x_i$ and $t_i$. The updated inverse Bellman error, when evaluated at the arbitrarily selected state and at time $t_i$ using the estimates of the model and the optimal controller, is given by
	\begin{equation}
	\!\delta^{\prime \prime}\!\left(t_i,{x}_i,\hat{u}_i,\hat{W}^{\prime}(t_i),\hat{\theta}(t_i)\right)\!=\!\left(\hat{W}^{\prime}(t_i)\right)^{T}\!\!\sigma^{\prime}\!\left(t_i,{x}_i,\hat{u}_i,\hat{\theta}(t_i)\right)\!,\label{eq:inverse BE Control Estimates}
	\end{equation}where 
	\begin{equation*}
		 \hat{W}^{\prime}(t_i) \coloneqq \left[\hat{W}_{V}(t_i);\hat{W}_{Q}(t_i);\hat{W}_{R}(t_i)\right]
	\end{equation*}
	and \begin{multline*}
	\sigma^{\prime}\left(t_i,{x}_i,\hat{u}_i,\hat{\theta}(t_i)\right)\coloneqq\Big[\Big(\left[\nabla_{x}\sigma_{V}\right]\left({x}_i\right)\Big)(f^o(x_i,\hat{u}_i)\\+\hat{\theta}^T(t_i)\sigma(x_i,\hat{u}_i));\sigma_{Q}\left({x}_i\right);\sigma_{u}\left(\hat{u}_i\right)\Big]. 
	\end{multline*} 
	
%

	Since all positive multiples of a reward function result in the same optimal controller, given state-action pairs, the reward function can only be identified up to a scale. As a result, one of the reward function weights can be arbitrarily assigned. 
	
	Since optimal control behaviors are scale-invariant, there is no loss of generality in resolving the scale ambiguity by taking the first element of $ \hat{W}_{R} $ to be known. The inverse BE in \eqref{eq:inverse BE Control Estimates} can then be expressed as \begin{multline}
	\!\!\!\!\!\delta^{\prime \prime}\!\left(t_i,{x}_i,\hat{u}_i,\hat{W}(t_i),\hat{\theta}(t_i)\right)=\left(\hat{W}(t_i)\right)^{T}\!\!\sigma^{\prime\prime}\!\left(t_i,{x}_i,\hat{u}_i,\hat{\theta}(t_i)\right) \\+ r_{1}\sigma_{u1}\left(\hat{u}_i\right),\label{eq: Deltapp}
	\end{multline}where $ \hat{W}(t_i) \coloneqq \left[\hat{W}_{V}(t_i);\hat{W}_{Q}(t_i);\hat{W}_{R}^{-}(t_i)\right],$ the vector $ \hat{W}_{R}^{-} $ denotes $ \hat{W}_{R} $ with the first element removed, $\sigma_{uj}\left(\hat{u}_i\right) $ denotes the $ j $th element of the vector $ \sigma_{u}\left(\hat{u}_i\right) $, the vector $ \sigma_{u}^{-} $ denotes $ \sigma_{u} $ with the first element removed, and 
	\begin{multline}
		\sigma^{\prime\prime}\left(t_i,{x}_i,\hat{u}_i,\hat{\theta}(t_i)\right)\coloneqq\Big[\Big(\left[\nabla_{x}\sigma_{V}\right]\!\left({x}_i\right)\Big)(f^o(x_i,\hat{u}_i)\\+\hat{\theta}^T(t_i)\sigma(x_i,\hat{u}_i));\sigma_{Q}\left({x}_i\right);\sigma_{u}^{-}\left(\hat{u}_i\right)\Big] .
	\end{multline}
	
	The closed-form nonlinear optimal controller corresponding to the reward structure in (\ref{eq:reward}) provides the relationship\begin{multline}
	-2R{u}^*\left(x_i\right)=\Big(\left[\nabla_uf\right]\left(x_i\right)\Big)^T\Big(\left[\nabla_{x}\sigma_{V}\right]\left(x_i\right)\Big)^TW_{V}^{*}\\+\Big(\left[\nabla_uf\right]\left(x_i\right)\Big)^T\Big(\left[\nabla_{x}\epsilon_V\right]\left(x_i\right)\Big)^T.\label{eq: Optimal Control}
	\end{multline}
	
	Utilizing estimates $\hat{\theta}(t_i)$ and data pairs $\left(x_i,\hat{u}_i\right)$ in \eqref{eq: Optimal Control}, subtracting $H\left(x_i,\Big(\left[\nabla_xV\right](x_i)\Big),u^*(x_i)\right)$ from \eqref{eq: Deltapp}, evaluating \eqref{eq: Deltapp} and \eqref{eq: Optimal Control} at time instances $\{t_i\}_{i=1}^N$, and stacking the results in a matrix form, we get
	\begin{equation}
	-\hat{\Sigma}\hat{W}-\hat{\Sigma}_{u1}=\hat{\Sigma}\tilde{W}-\Delta, \label{eq: Weight Linear System}
	\end{equation}where the weight estimation error is defined as $\tilde{W}=W^*-\hat{W}$, and $\hat{W}$ is the estimate of $W^*$, and
	\begin{align*}
	&\hat{\Sigma}\coloneqq\!\left[\sigma^T\!\left(t_1,{x}_1,\hat{u}_1,\hat{\theta}\left(t_1\right)\right)\!;\cdots;\sigma^T\!\left(t_N,{x}_N,\hat{u}_N,\hat{\theta}\left(t_N\right)\right)\right]\!\!,\\ &\hat{\Sigma}_{u1}\coloneqq\left[\sigma_{u1}^{\prime}\left(\hat{u}_1\right);\cdots;\sigma_{u1}^{\prime}\left(\hat{u}_N\right)\right],\\
	&\Delta:=\Big[\Delta_\delta(t_1);\Delta_m(t_1);\cdots;\Delta_\delta(t_N);\Delta_m(t_N)\Big],
	\end{align*}
	where
	\begin{equation*}
		\sigma_{u1}^{\prime}(\hat{u}_i)\coloneqq \left[r_1\sigma_{u1}\left(\hat{u}_{1i}\right);2r_{1}\hat{u}_{1i};0_{\left(m-1\right)\times 1}\right],
	\end{equation*}
	\begin{equation*} \sigma\coloneqq\!\left[\sigma^{\prime\prime} \!\begin{bmatrix}
		G\\\left[0_{m\times L}, \,\,\begin{bmatrix}
		0_{1\times m-1}\\2\text{diag}\left(\left[\hat{u}_{2i},\cdots,\hat{u}_{mi}\right]\right)
		\end{bmatrix}\right]^{T}
		\end{bmatrix}\right],
		\end{equation*}
		\begin{equation*}
			G:=\!\Big(\!\left[\nabla_{x}\sigma_V\right](x_i)\Big)\Big(\!\Big(\!\left[\nabla_uf^o\right](x_i)\Big)+\hat{\theta}^T(t_i)\Big(\!\left[\nabla_u\sigma\right](x_i)\Big)\!\Big),
		\end{equation*}
		\begin{multline*}
			\!\!\!\!\!\Delta_\delta(t_i):=2R\tilde{u}_i+\Big(\!\left[\nabla_u\sigma\right](x_i)\Big)^T\tilde{\theta}(t_i)\Big(\!\left[\nabla_u\sigma_V\right](x_i)\Big)^TW_V^*\\+\left(\!\Big(\left[\nabla_uf^o\right](x_i)\!\Big)\!\!+\theta^T(t_i)\Big(\!\left[\nabla_u\sigma\right](x_i)\Big)\!\right)^T\!\Big(\!\left[\nabla_x\epsilon_V\right](x_i)\!\Big)^T\\+\Big(\left[\nabla_u\epsilon\right](x_i,u^*_i)\Big)\Big(\left[\nabla_x\sigma_V\right](x_i)\Big)W_V^*,
		\end{multline*}
		\begin{multline*}
			\Delta_m(t_i):=\Big(\sigma_u(u_i^*)-\sigma_u(\hat{u}_i)\Big)^TW_R^*+\epsilon_V(x_i)+\epsilon_Q(x_i)\\+\Big({f^{o}}\left(x_i,u_i^*\right)-{f^{o}}\left(x_i,\hat{u}_i\right)\Big)^T\Big(\left[\nabla_x\sigma_V\right](x_i)\Big)^TW_V^*\\+\left(\theta^T\left(\sigma(x_i,u^*_i)-\sigma(x_i,\hat{u}_i)\right)\right)^T\Big(\left[\nabla_x\sigma_V\right](x_i)\Big)^TW_V^*\\+\left(\tilde{\theta}^T(t_i)\sigma(x_i,\hat{u}_i)+\epsilon(x_i,u^*_i)\right)^T\Big(\left[\nabla_x\sigma_V\right](x_i)\Big)^TW_V^*,
		\end{multline*}
	and $\hat{u}_{ji}$ is the $j$th element of $\hat{u}_i$. 
	
	A history stack, denoted as $\mathcal{H}^{IRL}$, is a set of ordered pairs of parameter estimates, $\hat{\theta}(t_i)$, and data pairs, $(x_i,\hat{u}_i)$, collected over time instance $t_1,t_2,\ldots,t_N$ into matrices $\Big(\hat{\Sigma},\hat{\Sigma}_{u1}\Big)$.
		
	Due to the fact that $\hat{\Sigma}$ and $\Delta$ depend on the quality of the control and parameter estimates, a purging technique is incorporated in the following to remove poor estimates $\hat{u}$ and $\hat{\theta}$ from $\mathcal{H}^{IRL}.$ During the transient phase of the control and parameter estimators, the estimates $\hat{u}$ and $\hat{\theta}$ are likely to be less accurate and the resulting values of $\hat{W}$ are likely to be poor. Purging facilitates usage of better estimates as they become available.
	
	The recursive update law is then designed as
	\begin{equation}
	\dot{\hat{W}}=\alpha\Gamma\hat{\Sigma}^T\left(-\hat{\Sigma}\hat{W}-\hat{\Sigma}_{u1}\right).\label{eq:W Dynamics}
	\end{equation}
	In \eqref{eq:W Dynamics}, $\alpha\in\mathbb{R}_{>0}$ is a constant adaptation gain and $\Gamma:\mathbb{R}_{\geq0}\to\mathbb{R}^{\left(L+P+m-1\right)\times \left(L+P+m-1\right)}$ is the least-squares gain updated using the update law
	\begin{equation}
	\dot{\Gamma}=\beta\Gamma-\alpha\Gamma\hat{\Sigma}^T\hat{\Sigma}\Gamma,\label{eq:WGamma Dynamics}
	\end{equation}
	where $\beta\in\mathbb{R}_{>0}$ is the forgetting factor. 
	
	\subsection{Analysis}\label{sec:StabilityNoState}
	
	A Lyapunov based analysis is performed to show convergence for the IRL method in Section \ref{sec:IRLNoState}.
	
	The time-varying history stack, $\mathcal{H}^{IRL}$, is called full rank, uniformly in $t$, if there exists a $\underline{\sigma}>0$ such that $\forall t\in\mathbb{R}_{\geq T_0}$,
	\begin{equation}
	0<\underline{\sigma}<\lambda_{\min}\left\{\hat{\Sigma}^T(t)\hat{\Sigma}(t)\right\}.\label{def: Full Rank}
	\end{equation} 
	
	Using arguments similar to \cite[Corollary 4.3.2]{SCC.Ioannou.Sun1996}, it can be shown that if $\lambda_{\min}\left\{ \Gamma^{-1}\left(T_0\right)\right\} >0$, and if $\mathcal{H}^{IRL}$ is full rank, uniformly in $t$, then the least squares gain matrix satisfies
	\begin{equation}
	\underline{\Gamma}\text{I}_{L+P+m-1}\leq\Gamma\left(t\right)\leq\overline{\Gamma}\text{I}_{L+P+m-1},\label{eq:NLSStaFGammaBoundNoState}
	\end{equation}
	where $\underline{\Gamma}$ and $\overline{\Gamma}$ are positive
	constants.
	
	To facilitate the following Lyapunov analysis, using \eqref{eq:W Dynamics}, the dynamics for the weight estimation error can be described by
	\begin{equation}
	\dot{\tilde{W}}=-\alpha\Gamma\hat{\Sigma}^T\left(\hat{\Sigma}\tilde{W}-\Delta\right).\label{eq:Brunovsky Werrordyns}
	\end{equation} 
	
	The stability result is summarized in the following theorem.
	\begin{thm}\label{thm: Brunovsky UUB W}
		If $\mathcal{H}^{IRL}$ is full rank, uniformly in $t$, then $t \mapsto \tilde{W}\left(t\right)$ is uniformly ultimately bounded.
	\end{thm}	
	\begin{proof}
		Consider the positive definite candidate Lyapunov function
		\begin{equation}
		V(\tilde{W},t)=\frac{1}{2}\tilde{W}^T\Gamma^{-1}\left(t\right)\tilde{W}. \label{eqn:Brunovsky V_estimation}
		\end{equation}
		Using the bounds in \eqref{eq:NLSStaFGammaBoundNoState}, the candidate Lyapunov function satisfies 
		\begin{equation}
		\frac{1}{2\overline{\Gamma}}\left\Vert \tilde{W}\right\Vert ^{2}\leq V\left(\tilde{W},t\right)\leq\frac{1}{2\underline{\Gamma}}\left\Vert \tilde{W}\right\Vert ^{2}.\label{eq:Brunovsky NLSCLNoXDotVBounds-1}
		\end{equation}
		Taking the time-derivative of \eqref{eqn:Brunovsky V_estimation}, and 
		using \eqref{eq:WGamma Dynamics} and \eqref{eq:Brunovsky Werrordyns}, along with the identity $\dot{\Gamma}^{-1}=-\Gamma^{-1}\dot{\Gamma}\Gamma^{-1}$, after simplifying the time-derivative can be expressed as
		\begin{multline}
		\dot{V}(\tilde{W},t)=-\frac{1}{2}\alpha\tilde{W}^T\hat{\Sigma}^T\hat{\Sigma}\tilde{W}+\alpha\tilde{W}^T\hat{\Sigma}^T\Delta\\-\frac{1}{2}\beta\tilde{W}^T\Gamma^{-1}\left(t\right)\tilde{W}.
		\end{multline}
		Substituting in $\hat{\Sigma}=\Sigma-\tilde{\Sigma}$, and using the Cauchy-Schwartz inequality and bounds in \eqref{def: Full Rank} and (\ref{eq:NLSStaFGammaBoundNoState}), $\dot{V}$ can be bounded by
		\begin{multline}
		\label{eq:Brunovsky Dot_V_estimation_bound2}
		\dot{V}(\tilde{W},t)\leq -\frac{1}{2}\left(\alpha\underline{\sigma}+\frac{1}{\overline{\Gamma}}\beta\right)\left\Vert\tilde{W}\right\Vert^2+\alpha\Vert\tilde{W}\Vert\Vert\Sigma\Vert\Vert\Delta\Vert\\+\alpha\Vert\tilde{W}\Vert\Vert\tilde{\Sigma}\Vert\Vert\Delta\Vert.
		\end{multline}
		\begin{remark}\label{rem: Lipschitz}
			Since $\left(x,u\right)\mapsto f\left(x,u\right)$, $\left(x,u\right)\mapsto\sigma\left(x,u\right)$, $u\mapsto \sigma_u\left(u\right)$, and $x\mapsto \sigma_V\left(x\right)$ are continuously differentiable,
			and since $t\mapsto u\left(t\right)$ is bounded,
			given a compact set $\hat{\chi}\subset\mathbb{R}^n\times\mathbb{R}^m\times\mathbb{R}^{m}$,
			there exist $L_{\sigma1},L_{F1},L_{R1}>0$ 
			such that \begin{align}
			\sup_{\left(x,u,\hat{u}\right)\in\hat{\chi}}\left\Vert \tilde{\sigma}\left(x,u,\hat{u}\right)\right\Vert &\leq L_{\sigma 1}\left\Vert \tilde{u}\right\Vert, \nonumber\\
			\sup_{\left(x,u,\hat{u}\right)\in\hat{\chi}}\left\Vert \tilde{f^{o}}\left(x,u,\hat{u}\right)\right\Vert &\leq L_{F1}\left\Vert \tilde{u}\right\Vert, \nonumber\\
			\sup_{\left(x,u,\hat{u}\right)\in\hat{\chi}}\left\Vert \tilde{\sigma}_u\left(u,\hat{u}\right)\right\Vert &\leq L_{R1}\left\Vert \tilde{u}\right\Vert.
			\end{align}
		\end{remark}
		Using Remark \ref{rem: Lipschitz}, the term $\left\Vert\tilde{\Sigma}\right\Vert$ can be expressed in terms of $\tilde{u}$ and $\tilde{\theta}$ as
		\begin{equation}
		\left\Vert\tilde{\Sigma}\right\Vert\leq\left(\left\Vert\tilde{u}\right\Vert+\left\Vert\tilde{\theta}\right\Vert\right) \overline{\Sigma}, \label{eq: Brunovsky Sigma Tilde Bound}
		\end{equation}
		where
		\begin{multline}
		\overline{\Sigma}:=N\sup_{\substack{\left(x,u,\hat{u}\right)\in \hat{\chi}}}\Big\{\Vert\nabla_x\sigma_V(x)\Vert \Big(L_{F1}+L_{\sigma 1}(\Vert\tilde{u}\Vert+\Vert\tilde{\theta}\Vert)\\+L_{\sigma 1}\Vert\theta\Vert+\Vert\sigma(x,u)\Vert+\Vert\nabla_u\sigma(x,u)\Vert\Big),2+L_{R1}\Big\}.
		\end{multline}
		The term $\left\Vert\Sigma\right\Vert$, which contains true values of the unknown states and parameters, is bounded above since it is a function of only true controls and parameters, $u$ and $\theta$, and queried states $x_i$. Let the upper bound on $\left\Vert\Sigma\right\Vert$ be denoted as
		\begin{equation}\label{eq: Brunovsky Sigma Bound}
		\left\Vert\Sigma\right\Vert\leq\overline{\Sigma}_\sigma,
		\end{equation}
		where
		\begin{align}
		\overline{\Sigma}&_\sigma:=N\sup_{\substack{ x \in x\left(\cdot\right) \\ u \in u\left(\cdot\right)}}\Big\{\left\Vert\nabla_x\sigma_{V}(x)\right\Vert\Big(\left\Vert f^o(x,u)\right\Vert+\left\Vert\theta\right\Vert\left\Vert\sigma(x,u)\right\Vert\nonumber\\&+\left\Vert\nabla_uf^o(x,u)\right\Vert+\left\Vert\theta\right\Vert\left\Vert\nabla_u\sigma(x,u)\right\Vert\Big),\left\Vert\sigma_{u}^-\left(u\right)\right\Vert,\nonumber\\&\left\Vert\nabla_u\sigma_u^-\left(u\right)\right\Vert\Big\}.
		\end{align}
		The error term $\Vert\Delta\Vert$ is bounded above by
		\begin{equation}
			\Vert\Delta\Vert\leq\left(\Vert\tilde{u}\Vert+\Vert\tilde{\theta}\Vert\right)\overline{\Delta}+\overline{\Delta}_\epsilon,\label{eq: error bound}
		\end{equation}
		where 
		\begin{multline}
			\overline{\Delta}:=N\sup_{\substack{\left(x,u,\hat{u}\right)\in \hat{\chi}}}\Big\{L_{R1}\Vert W_R^*\Vert+2\Vert R\Vert\\+\Vert\left[\nabla_u\sigma\right](x)\Vert\Vert\left[\nabla_u\sigma_V\right](x)\Vert\Vert W_V^*\Vert\\+\Vert\sigma(x,\hat{u})\Vert\Vert\left[\nabla_x\sigma_V\right](x)\Vert\Vert W_V^*\Vert\\+L_{F1}\Vert\left[\nabla_x\sigma_V\right](x)\Vert\Vert W_V^*\Vert\\+L_{\sigma1}\Vert\theta^T\Vert\Vert\left[\nabla_x\sigma_V\right](x)\Vert\Vert W_V^*\Vert\Big\},
		\end{multline}
		and 
		\begin{multline}
			\overline{\Delta}_\epsilon:=N\sup_{\substack{\left(x,u,\hat{u}\right)\in \hat{\chi}}}\Big\{\Vert\epsilon_V(x)\Vert+\Vert\epsilon_Q(x)\Vert\\+\Vert\epsilon(x,u^*)\Vert\Vert\left[\nabla_x\sigma_V\right](x)\Vert\Vert W_V^*\Vert\\+\Big(\Vert\left[\nabla_uf^o\right](x)\Vert+\Vert\theta\Vert\Vert\left[\nabla_u\sigma\right](x)\Vert\Big)\Vert\left[\nabla_x\epsilon_V\right](x)\Vert\\+\Vert\left[\nabla_u\epsilon\right](x,u^*)\Vert\Vert\left[\nabla_x\sigma_V\right](x)\Vert\Vert W_V^*\Vert\Big\}.
		\end{multline}
		Using \eqref{eq: Brunovsky Sigma Tilde Bound}, \eqref{eq: Brunovsky Sigma Bound} and \eqref{eq: error bound}, $\dot{V}$ becomes
		\begin{multline}
		\!\!\dot{V}(\tilde{W},t)\leq -\frac{1}{2}\left(\alpha\underline{\sigma}+\frac{1}{\overline{\Gamma}}\beta\right)\left\Vert\tilde{W}\right\Vert^2+\alpha\overline{\Delta}_\epsilon\overline{\Sigma}_\sigma\left\Vert\tilde{W}\right\Vert\\+\alpha\overline{\Delta}_\epsilon\overline{\Sigma}\left\Vert\tilde{W}\right\Vert\!\left(\!\left\Vert\tilde{u}\right\Vert\!+\!\left\Vert\tilde{\theta}\right\Vert\right)+\alpha\overline{\Sigma}_\sigma\overline{\Delta}\left\Vert\tilde{W}\right\Vert\!\left(\!\left\Vert\tilde{u}\right\Vert\!+\!\left\Vert\tilde{\theta}\right\Vert\right)\\+\alpha\overline{\Sigma} \ \overline{\Delta}\left\Vert\tilde{W}\right\Vert\left(\left\Vert\tilde{u}\right\Vert+\left\Vert\tilde{\theta}\right\Vert\right)^2.
		\end{multline}
		Using Young's Inequality $\dot{V}$ then becomes
		\begin{multline}
		\label{eq:Brunovsky Dot_V_estimation_Final bound}
		\!\!\dot{V}(\tilde{W},t)\leq  -\frac{1}{8}\left(\alpha\underline{\sigma}+\frac{1}{\overline{\Gamma}}\beta\right)\left\Vert\tilde{W}\right\Vert^2+\frac{2\alpha^2\overline{\Delta}_\epsilon^2\overline{\Sigma}_\sigma^2}{\alpha\underline{\sigma}+\nicefrac{\beta}{\overline{\Gamma}}}\\+\frac{\alpha^2\overline{\mathcal{E}}^2\left(\overline{\Delta}_\epsilon\overline{\Sigma}+\overline{\Sigma}_\sigma\overline{\Delta}+\overline{\Sigma} \ \overline{\Delta} \ \overline{{\mathcal{E}}}\right)^2}{\alpha\underline{\sigma}+\nicefrac{\beta}{\overline{\Gamma}}},
		\end{multline}
		where $\overline{\mathcal{E}} = \left(\overline{\tilde{u}}+\overline{\tilde{\theta}}\right)$. The notation, $\overline{\tilde{u}}$ and $\overline{\tilde{\theta}}$, denote bounded $\tilde{u}$ and $\tilde{\theta}$ values stored in the history stack, $\mathcal{H}^{IRL}$.
		Using the bound in \eqref{eq:Brunovsky NLSCLNoXDotVBounds-1}, the differential inequality for $\dot{V}$ can be expressed as
		\begin{equation}
		\label{eq: Brunovsky Differential V of Z}
		\!\!\dot{V}(\tilde{W},t)\leq -CV\left(\tilde{W},t\right)+D,
		\end{equation}
		where
		\begin{equation}
		C:=\frac{\underline{\Gamma}}{4}\left(\alpha\underline{\sigma}+\frac{1}{\overline{\Gamma}}\beta\right),
		\end{equation}
		\begin{multline}
		D:=\frac{\alpha^2\overline{\mathcal{E}}^2\left(\overline{\Delta}_\epsilon\overline{\Sigma}+\overline{\Sigma}_\sigma\overline{\Delta}+\overline{\Sigma} \ \overline{\Delta} \ \overline{{\mathcal{E}}}\right)^2}{\alpha\underline{\sigma}+\nicefrac{\beta}{\overline{\Gamma}}}+\frac{2\alpha^2\overline{\Delta}_\epsilon^2\overline{\Sigma}_\sigma^2}{\alpha\underline{\sigma}+\nicefrac{\beta}{\overline{\Gamma}}}.
		\end{multline}
		
		Due to purging of $\mathcal{H}^{IRL}$, the estimator is analyzed over discrete time instances. Define the purging instances as $T_1,T_2,\ldots$, and maintain a minimum dwell time, $\mathcal{T}$, such that $T_{s+1}-T_s\geq\mathcal{T} >0, \ \forall s\in\mathbb{N}$.
		
		Solving equation $\eqref{eq: Brunovsky Differential V of Z}$ over any time interval $[T_s,T_{s+1})$, yields
		\begin{equation}\label{eq: Brunovsky V Equation}
		\overline{V}_{s+1}\leq \overline{V}_se^{-C\left(t-T_s\right)}+\frac{D_{s+1}}{C},
		\end{equation}
		where $\overline{V}_s\geq\left\Vert V\left(\tilde{W}\left(T_{s}\right),T_{s}\right)\right\Vert$ and $D_{s+1}$ denotes the value of $D$ over interval $[T_s,T_{s+1})$.  Since we know that $\tilde{\theta}$ and $\tilde{u}$ decay exponentially to a bound, we know that $\overline{\mathcal{E}}$ is decreasing exponentially. Therefore, due to the decreasing error term $\overline{\mathcal{E}}$, it can be seen that
		\begin{equation}\label{eq: Brunovsky Decreasing B}
		D_s>D_{s+1}, \forall s = 1,2,\ldots
		\end{equation}
		and \begin{multline}
		\overline{D}:=\lim\limits_{s\to\infty}D_s=\frac{2\alpha^2\overline{\Delta}_\epsilon^2\overline{\Sigma}_\sigma^2}{\alpha\underline{\sigma}+\nicefrac{\beta}{\overline{\Gamma}}}\\+\frac{\alpha^2\overline{\mathcal{E}}_N^2\left(\overline{\Delta}_\epsilon\overline{\Sigma}+\overline{\Sigma}_\sigma\overline{\Delta}+\overline{\Sigma} \ \overline{\Delta} \ \overline{{\mathcal{E}}}_N\right)^2}{\alpha\underline{\sigma}+\nicefrac{\beta}{\overline{\Gamma}}},\end{multline} where $\overline{\mathcal{E}}_N:=\left(\sqrt{\overline{\Gamma}_u\frac{{B}}{A}}+\overline{\theta}_\infty\right)$, and $\overline{\theta}_\infty$ denotes the ultimate bound of the parameter estimation error $\tilde{\theta}$. Furthermore, the dwell time condition results in the bound
		\begin{align}
		\overline{V}_{s+1}\leq \overline{V}_se^{-C\mathcal{T}}+\frac{D_{s+1}}{C}, \forall s = 0,1,2,\ldots
		\end{align}
		If the bounds $D_{s+1}$ are selected so that 
		\begin{equation}\label{eq:Brunovsky ref}
		D_{s+1}>2D_se^{-C\mathcal{T}},\forall s = 0,1,2, \ldots,
		\end{equation} then
		\begin{equation}
		\overline{V}_{s+1}\leq \frac{2D_{s+1}}{C}, \forall s = 0,1,2,\ldots,
		\end{equation}
		where $D_0:=\frac{C\overline{V_0}}{2}$. As a result, it can be concluded that
		\begin{equation}
		\lim\limits_{s \to \infty}\sup\overline{V}_s\leq\frac{2\overline{D}}{C},
		\end{equation}and as a result  $\lim\limits_{s\to\infty}\sup\left\Vert\tilde{W}\left(T_s\right)\right\Vert\leq2\sqrt{\overline{\Gamma}\frac{\overline{D}}{C}}$.
		
	\end{proof}

	\section{Simulation}\label{sec: Sim}
	To demonstrate the performance of the developed method,
	a linear optimal trajectory tracking problem, using the method developed in \cite{SCC.Kamalapurkar.Dinh.ea2015,SCC.Kamalapurkar.Andrews.ea2017}, is utilized
	in order to have a known value function for comparison.
	
	Consider an agent with the following linear dynamics
	\begin{equation}
		\dot{x}=\begin{bmatrix}
		0 & 1\\
		\theta_1 & \theta_2
		\end{bmatrix}+\begin{bmatrix}
		0\\\theta_3
		\end{bmatrix}u,
	\end{equation}
	where the unknown parameters are $\theta_1=-0.5,\theta_2=-0.5,$ and $\theta_3=1$. The parameter estimation technique utilized is developed in \cite{SCC.Kamalapurkar2017}. 
	
	The trajectory the agent is attempting to follow is generated from the linear system
	\begin{equation}
		\dot{x}_d = \begin{bmatrix}
		0 & 1\\
		-2 & 0
		\end{bmatrix}x_d.
	\end{equation}
	The optimal control problem designed on the error dynamics is
	\begin{equation}
		J(e_0,\mu(\cdot)) = \int_{T_0}^{\infty}e(t)^T\begin{bmatrix}
		1 & 0\\
		0 & 1
		\end{bmatrix}e(t)+10\mu(t)^2\dif t,
	\end{equation}
	resulting in the ideal reward function weights $Q =\text{diag}([W_{Q_1}, \ W_{Q_2}])= \text{diag}([1, \ 1])$ and $R = 10$ where the error dynamics are
	\begin{equation}
		\dot{e}=\begin{bmatrix}
		0 & 1\\
		-0.5 & -0.5
		\end{bmatrix}e+\begin{bmatrix}
		0\\1
		\end{bmatrix}\mu,
	\end{equation}
	where $e=x-x_d$, $\mu=u-u_d$, and $u_d= [ -1.5, \ 0.5] \ x_d$. The optimal value function to be estimated is
	\begin{equation}
		V^*=W_{V_1}e_1^2+W_{V_2}e_2^2+W_{V_3}e_1e_2,
	\end{equation}
	where the ideal values are $W_{V_1}=1.82,W_{V_2}=2.30,$ and $W_{V_3}=1.83$. The optimal controller is $\mu=-[0.092, \	0.230]e$.
	
	Fig. \ref{fig: Tracking error} shows the tracking error and Fig. \ref{fig: Parameter error} shows the parameter estimation error. The parameters used for the two simulations are: $\beta=0.5,\alpha=\nicefrac{0.01}{50},\beta_u = 2, \alpha_u = 1, M=50,N=50$ and a step size of $0.005s.$
	
	\begin{figure}
		\includegraphics[width=1\columnwidth]{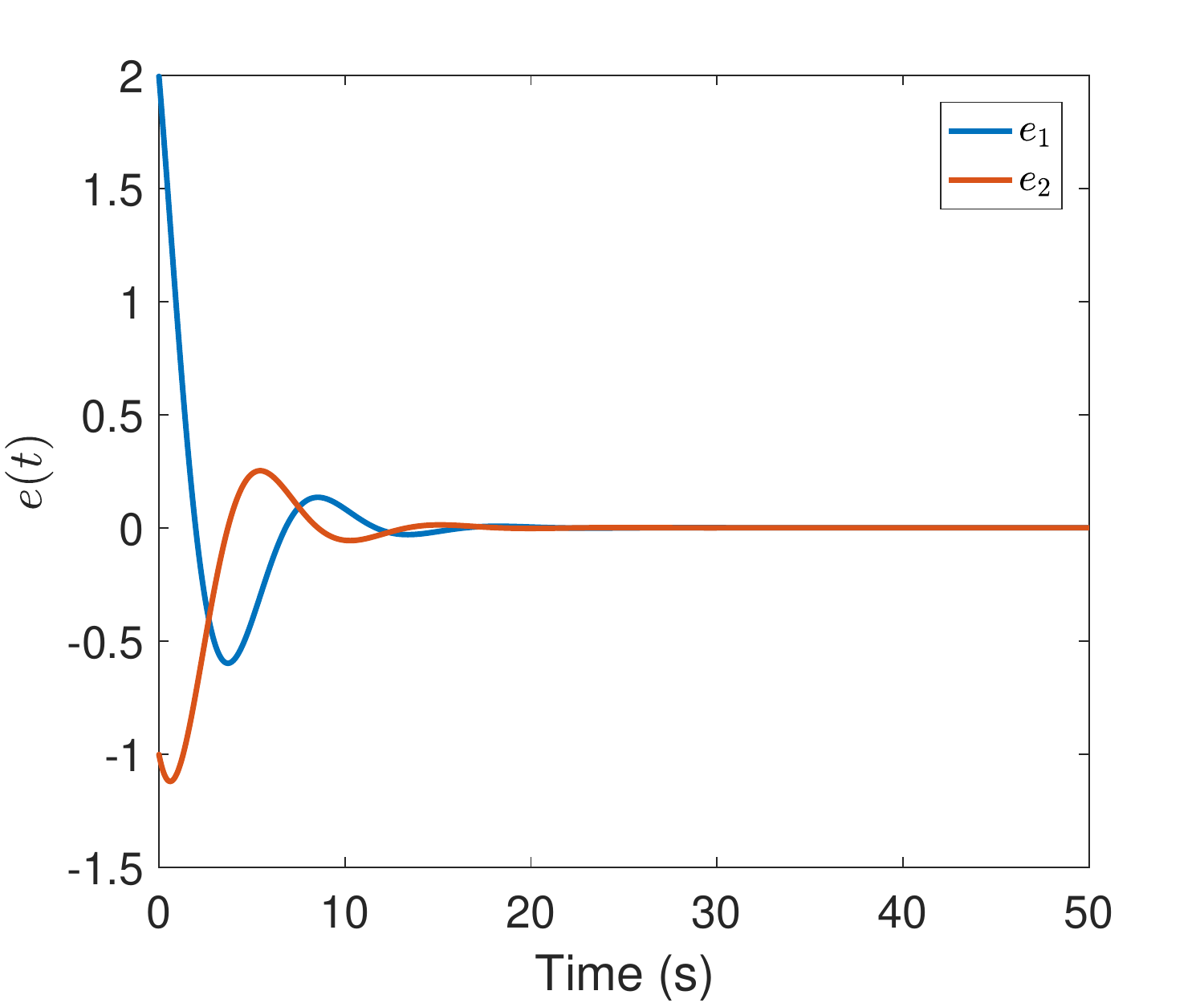}
		\caption{Trajectory tracking error.}
		\label{fig: Tracking error}
	\end{figure}
	
	\begin{figure}
		\includegraphics[width=1\columnwidth]{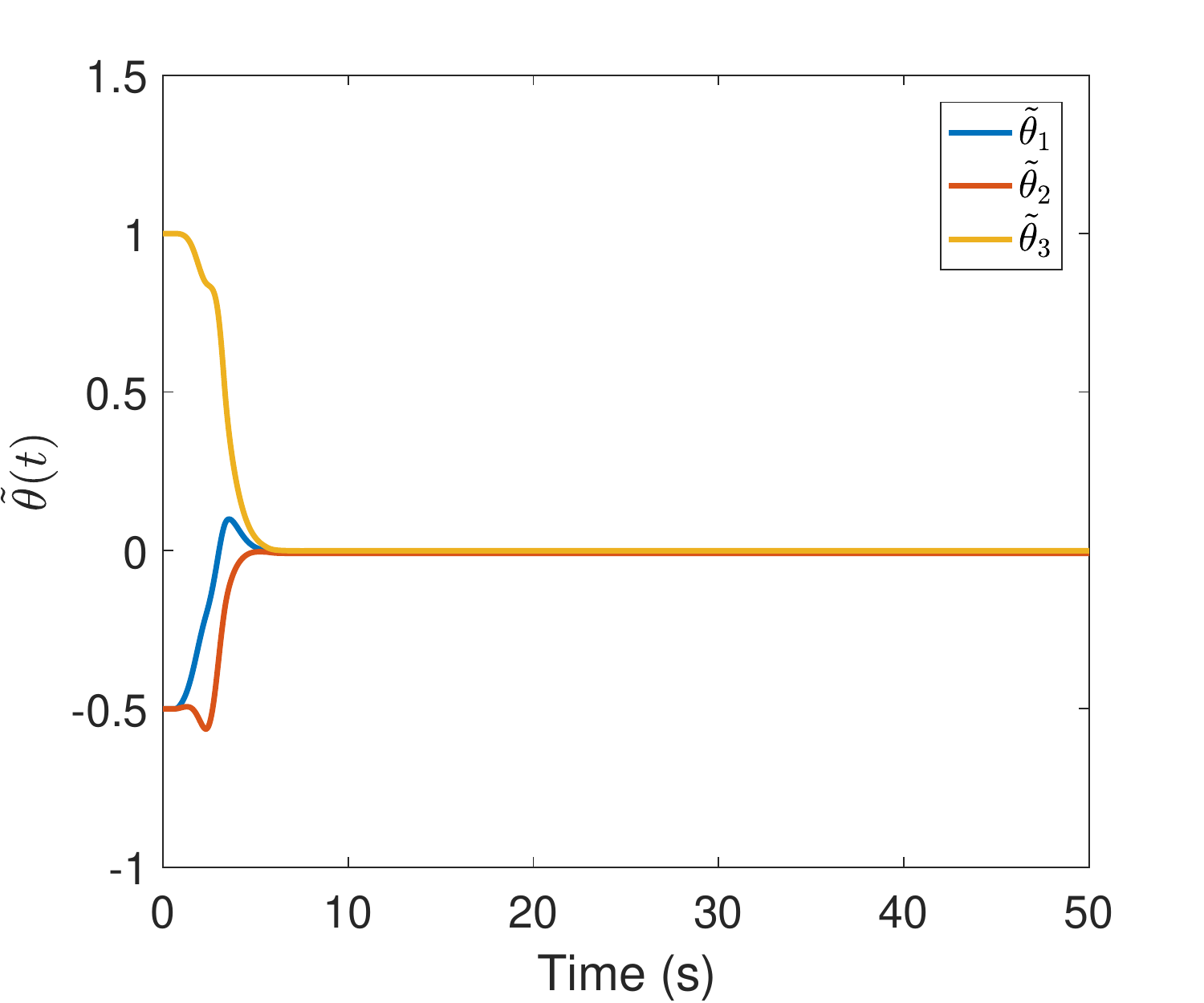}
		\caption{Parameter estimation error.}
		\label{fig: Parameter error}
	\end{figure}

	
	\subsection{IRL without Data Querying}
	
	The first simulation utilizes the state and control trajectories directly for IRL, and does not estimate the optimal controller for additional data. Fig. \ref{fig: Reward function without Wu} shows reward and value function estimation errors without queried data.
	
	As seen in Fig. \ref{fig: Reward function without Wu}, the reward and value function estimates do not converge to the ideal values. Looking closer, the estimates do not change much at all. The reason for this is once the history stacks are purged to remove poor parameter estimates, $\hat{\theta}$, the tracking errors, $e$, have decreased near the origin. Meaning, the data that IRL is utilizing, both $e$ and $\mu$, are at or near zero. This data does not provide sufficient information in order to accurately estimate the reward function. 
	
	\begin{figure}
		\includegraphics[width=1\columnwidth]{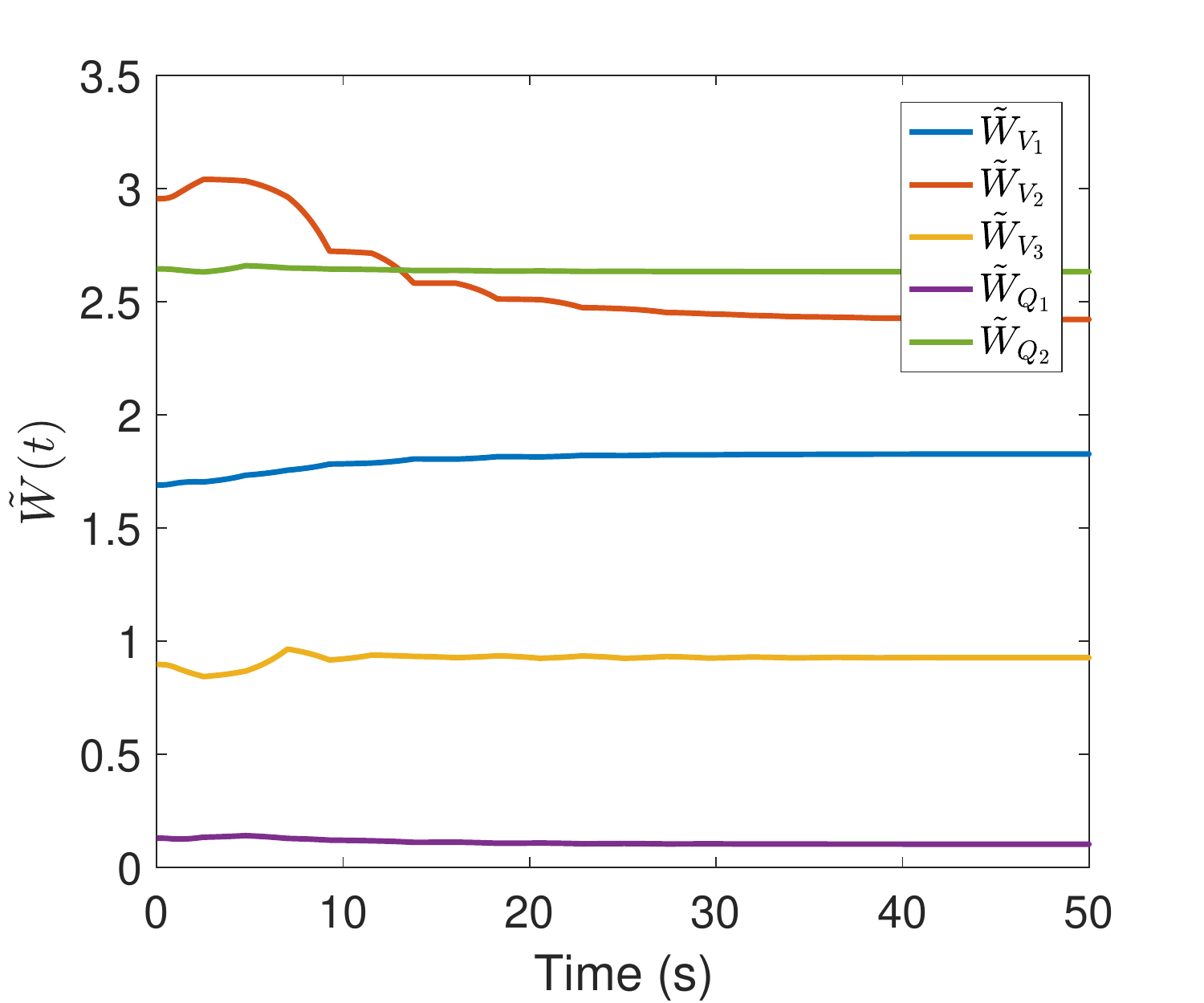}
		\caption{Reward and value function estimation error without data querying.}
		\label{fig: Reward function without Wu}
	\end{figure}
 
	\subsection{IRL Formulation with Data Querying}
	
	The second simulation shows the results of the novel control-estimation-based technique developed in this paper, with queried data points. 	
	\begin{figure}
		\includegraphics[width=1\columnwidth]{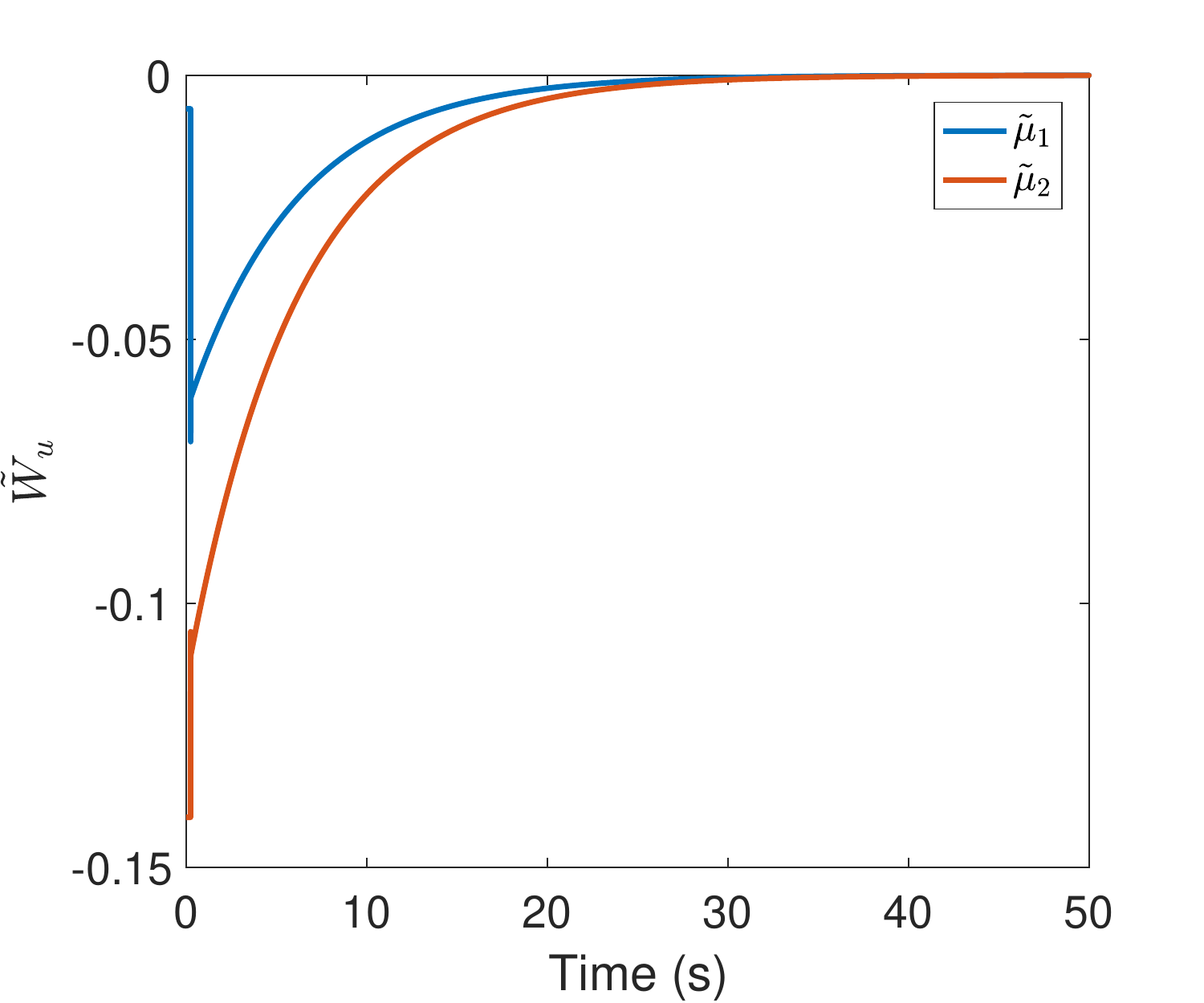}
		\caption{Optimal feedback controller estimation error.}
		\label{fig: Wu}
	\end{figure}	Utilizing the estimate of the optimal controller, the estimate is queried with random states $x_i$ in the set $[-1,1]$, which produce estimates of the optimal controller, $\hat{u}_i$. The pairs $(x_i,\hat{u}_i)$ are then iteratively collected in $\mathcal{H}^{IRL}$ and IRL is performed utilizing the update law in \eqref{eq:W Dynamics}.

\begin{figure}
	\includegraphics[width=1\columnwidth]{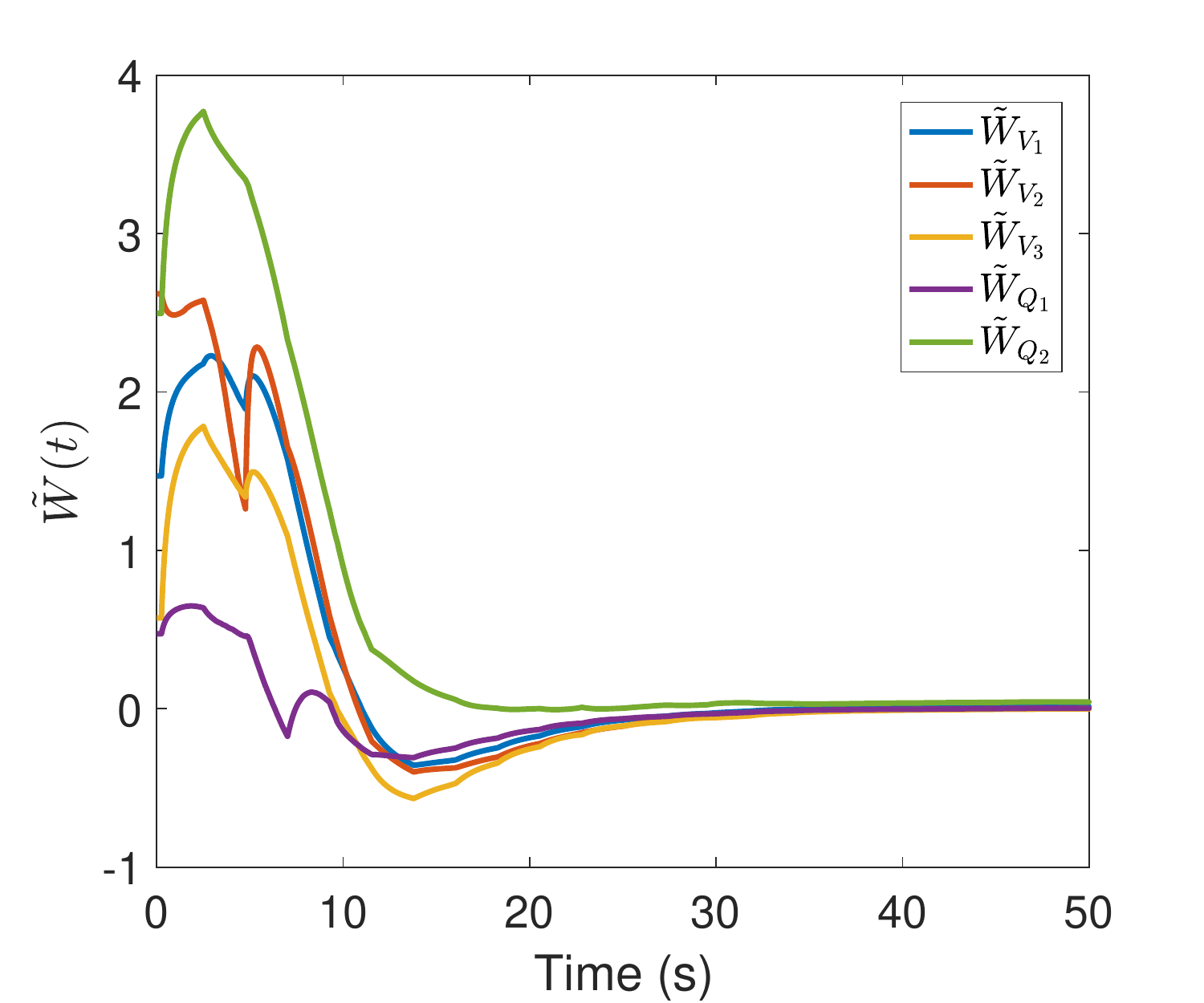}
	\caption{Reward and value function estimation error with data querying.}
	\label{fig: Reward function with Wu}
\end{figure}

	Fig. \ref{fig: Wu} shows the estimation error for the optimal feedback controller, and Fig. \ref{fig: Reward function with Wu} shows the reward and value function estimation errors.

	As seen in Fig. \ref{fig: Reward function with Wu}, the new IRL approach estimates the ideal values of the reward and value functions online. Though the tracking errors of the system dynamics have already converged, due to the non-zero queried state and control values available through feedback estimation, IRL is able to converge.

	\section{Conclusion}\label{sec: Conclusion}
	This paper presents a new approach to performing reward function estimation online for situations with limited data. The approach utilizes a concurrent learning update law to estimate the optimal feedback controller of the agent online. This estimate is then utilized to artificially create additional data to promote reward function estimation. Theoretical guarantees are provided showing uniform ultimate boundedness of the unknown reward and value functions estimation errors using Lyapunov theory. A simulation example is performed that clearly shows the benefit of the method and how this additional queried data helps promote reward function estimation.
	
	Future work will include analyzing the performance of this approach for systems with unmeasurable states and the affect of noise on optimal control estimation.
	\bibliographystyle{ieeetran}
	\bibliography{scc,sccmaster}
\end{document}